\documentclass[amssymb,amsfonts]{amsart}
\usepackage{amssymb, latexsym}
\usepackage[T1]{fontenc}
\usepackage[english]{babel}
\usepackage{svn}
\usepackage{tensor}
\usepackage{hyperref}

\theoremstyle{plain}
\newtheorem{thm}{Theorem}

\theoremstyle{definition}
\newtheorem{defn}[thm]{Definition}

\theoremstyle{remark}
\newtheorem{rmk}[thm]{Remark}

\newcommand*{\Real}{\ensuremath{\mathbb{R}}}

\newcommand*{\set}[2]{\ensuremath{\left\{~ #1 ~\middle|~ #2 ~\right\}}}
\makeatletter

\newcommand*{\Rmnum}[1]{\expandafter\@slowromancap\romannumeral #1@}
\makeatother

\SVN $Date: 2012-03-01 11:09:07 +0100 (Thu, 01 Mar 2012) $
\SVN $Rev: 27 $

\begin{document}
\title[Positive mass theorem in 2D]{A positive mass theorem for two
spatial dimensions}
\author[W. W.-Y. Wong]{Willie Wai-Yeung Wong}
\address{\'{E}cole Polytechnique F\'{e}d\'{e}rale de Lausanne,
Switzerland}
\thanks{The author is supported in part by the Swiss National Science
Foundation.}
\thanks{Version \texttt{rev\SVNRev} of \texttt{\SVNRawDate}}
\email{willie.wong@epfl.ch}
% 35 - PDE, 35L - Hyperbolic, 35M - mixed type, 35S - psDO paraDO
% 83 - Gravity, 83C - GR, 83C20 special solutions and symmetry classes
%   83C22 EinsteinMaxwell 83C35 grav waves 83C57 black holes 83C75
%   singularities and cosmic censorship 83C60 spinor methods 83C55
%   macroscopic interaction of gravitation with matter
% 58 - Global analysis, 58J - PDE on manifolds, 58J45 hyperbolic 58J05
%   elliptic 58J40 psDO 58J47 propagation singularities

\begin{abstract}
We observe that an analogue of the Positive Mass Theorem in the
time-symmetric case for three-space-time-dimensional general 
relativity follows trivially from the Gauss-Bonnet theorem. In this
case we also have that the spatial slice is diffeomorphic to
$\Real^2$.
\end{abstract}

\maketitle
%\tableofcontents

In this short note we consider Einstein's equation without
cosmological constant, that is
\[ R_{\mu\nu} - \frac12 Rg_{\mu\nu} = T_{\mu\nu}~, \]
on $1+2$ dimensional space-times. This theory has long been considered
as a toy model with possible applications to cosmic strings and domain
walls, or to quantum gravity. For a survey please refer to
\cite{Brown1988, Carlip1998} and references within. This low
dimensional theory is generally considered as un-interesting
\cite{Collas1977} due to the fact that Weyl curvature vanishes
identically in $(1+2)$ dimensions, a fact often interpreted in the
physics literature as the theory lacking gravitational degrees of
freedom. Furthermore, the theory does not reduce in a Newtonian limit
\cite{BaBuLa1986}: the exterior space-time to compact sources is
necessarily locally flat and is typically asymptotically conical
\cite{Carlip1998, DeJaHo1984, Deser1985}. 

For these space-times, by considering point-sources, it is revealed
\cite{DeJaHo1984} that the mass should be identified with the angle
defects near spatial infinity. For static space-times with spatial
sections diffeomorphic to $\Real^3$, it is also known that the
asymptotic mass can be related to the integral of scalar curvature on
the spatial slice, and hence under a dominant energy assumption must
be positive. 

The purpose of this note is to remark that the topological assumption
is unnecessary. 

Throughout we shall assume that $(\Sigma,g)$ is a complete
two-dimensional Riemannian manifold which represents a
\emph{time-symmetric} spatial slice in a three-dimensional Lorentzian
space-time $(M,\bar{g})$ (that is, trace of the second fundamental form 
of $\Sigma \hookrightarrow M$ vanishes identically; in other words, 
the slice is \emph{maximal}). We assume that the 
dominant energy condition holds for $\bar{g}$, and in particular the 
ambient Einstein tensor satisfies $\bar{G}_{\mu\nu}\xi^\mu\xi^\nu \geq
0$ for any time-like $\xi^\mu$. The Gauss equation then immediately 
implies that $g$ has non-negative scalar curvature. 

\begin{defn}\label{defncon}
A complete two-dimensional Riemannian manifold $(\Sigma,g)$ 
is said to be \emph{asymptotically conical} if there exists a compact
subset $K\subsetneq \Sigma$ where $\Sigma\setminus K$ has finitely
many connected components, and such that if $E$ is a connected
component of $\Sigma\setminus K$, there exists a diffeomorphism $\phi:
E\to (\Real^2\setminus \bar{B}(0,1))$ where in the Cartesian
coordinates on $\Real^2$ the line element satisfies
\[ \mathrm{d}s^2 - \left[\mathrm{d}x^2 + \mathrm{d}y^2 - \frac{1-P^2}{x^2 +
y^2} \left( x~\mathrm{d}y - y~\mathrm{d}x\right)^2\right] \in 
O_2\left((x^2 + y^2)^{-\epsilon}\right)\]
for some $\epsilon > 0$ and $P> 0$. The 
notation $f\in O_2(r^{-2\epsilon})$ is a shorthand for
\[ \left|f\right| + r\left|\partial f\right| + r^2\left|\partial^2
f\right| \leq Cr^{-2\epsilon}~. \]
\end{defn}

\begin{rmk}
The decay condition is sufficient to imply that the scalar curvature
$S$ of $g$ is integrable on $\Sigma$. Note that in polar coordinates
$x = r\cos\theta$ and $y=r\sin\theta$ the conical metric in the square
brackets can be written as the conical 
\[ \mathrm{d}r^2 + P^2 r^2 \mathrm{d}\theta^2 \]
where we see that $m = 2\pi (1-P)$ is the angle defect for parallel
transport around the tip of the cone. 
\end{rmk}

\begin{thm}\label{maintheorem}
If $(\Sigma,g)$ is a complete asymptotically conical two-dimensional
orientable Riemannian manifold with pointwise non-negative scalar 
curvature, then
$\Sigma$ is diffeomorphic to $\Real^2$ and $m = 2\pi (1-P)$ is
non-negative. If furthermore $m = 0$ then $(\Sigma,g)$ is isometric to
the Euclidean plane. 
\end{thm}

\begin{proof}
Enumerate from $1\ldots N$ the asymptotic ends $(E_i,g_i)$ with
diffeomorphisms $\phi_i$ and constant $P_i$. By the asymptotic
structure, for sufficiently large $R_i$ the curve $\gamma_i =
\set{\phi_i^{-1}(x,y)}{x^2 + y^2 = R_i^2}$
in the end $\Sigma_i$ has positive geodesic curvature, if we choose
the orientation so that the inward normal is toward the compact set
$K$. Let $\Sigma_0\supsetneq K$ denote the compact manifold with 
boundary in $\Sigma$ that is bounded by the $\gamma_i$. Applying
Gauss-Bonnet theorem, using the fact that the geodesic curvatures are
all signed and the scalar curvature is non-negative, we have that
$\Sigma_0$ has positive Euler characteristic. As $\Sigma_0$ is
orientable and connected, and has nonempty boundary, it must be
diffeomorphic to a disc. Hence $\Sigma$ has only one asymptotic end
and is diffeomorphic to $\Real^2$. For sufficiently large $R$ we let
$\Sigma_R$ denote the compact region bounded by $x^2 + y^2 = R^2$. 
Using the Gauss-Bonnet theorem
again, along with the decay properties of the metric, we see that 
$m$, the angle defect, is in fact given by 
\[ m = \lim_{R\to\infty} \frac12 \int_{\Sigma_R} S ~\mathrm{dvol}_g =
\frac12 \int_\Sigma S~\mathrm{dvol}_g~.\]
Hence $m$ is necessarily nonnegative, with equality to 0 only in the
case $S\equiv 0$. 
\end{proof}

\begin{rmk}
One can analogously define the ``quasilocal mass'' $m_\gamma$ associated
to $\gamma\subsetneq\Sigma$ a simple closed curve by letting 
$m_\gamma$ be the angle defect for parallel transport around $\gamma$.
Then it is easy to see the this quantity has a monotonicity property:
if $\gamma_1$ is to the ``outside'' of $\gamma_2$, let $\Sigma_{1,2}$
be the annular region bounded by the two curves, we must have
\[ m_{\gamma_1} - m_{\gamma_2} = \frac12 \int_{\Sigma_{1,2}}
S~\mathrm{dvol}_g~. \]
\end{rmk}

\begin{rmk}
It was pointed out to the author by Julien Cortier that some similar
considerations in the asymptotically hyperbolic case was mentioned by
Chru\'sciel and Herzlich; see Remark 3.1 in \cite{ChrHer2003}. 

Indeed, Theorem \ref{maintheorem} follows also from some more powerful
classical theorems in differential geometry. The topological
classification can be deduced from, e.g.\ Proposition 1.1 in
\cite{LiTam1991}. One can also deduce the theorem (with some work)
from Shiohama's Theorem A \cite{Shioha1985}. As shown above, however,
in the very restricted case considered in this note the desired result
can be obtained with much less machinery. 
\end{rmk}

\begin{rmk}
The author would also like to thank Gary Gibbons for pointing out that
a similar argument to the proof of Theorem \ref{maintheorem}
was already used by Comtet and Gibbons (see end of section
2 of \cite{ComGib1988}) to establish a positive mass condition on
cylindrical space-times about a cosmic string; the main difference is
that in the above theorem we contemplate, and rule out, the possibility 
of multiple asymptotic ends, as well as non-trivial topologies inside
a compact region.
\end{rmk}

\begin{rmk}
One can also ask about asymptotically cylindrical spaces, which can be
formally viewed as a limit of cones. Indeed, if
in Definition \ref{defncon} we replace the asymptotic condition 
\[ \mathrm{d}s^2 \to \mathrm{d}r^2 + P^2 r^2 \mathrm{d}\theta^2 \]
with
\[ \mathrm{d}s^2 \to \mathrm{d}r^2 + (P^2 r^2 + p^2)
\mathrm{d}\theta^2~,\]
then the limit $P= 0$ is no longer degenerate, and in fact corresponds
to a spatial slice that is cylindrical at the end. In this case,
however, the topological statement in Theorem \ref{maintheorem} is no
longer true: the standard cylinder
$\mathbb{S}^1\times\mathbb{R}^1$ is flat, is asymptotically
cylindrical, and has \emph{two} asymptotic ends. However, it is easily
checked using the same method of proof as Theorem \ref{maintheorem}
that this is the only multiple-ended asymptotically cylindrical
surface to support a non-negative scalar curvature. 
\end{rmk}

\bibliographystyle{amsalpha}
\bibliography{../bib_files/jabrefmaster.bib}

\end{document}